\documentclass[conference]{IEEEtran}
\IEEEoverridecommandlockouts

\usepackage{amsmath,amssymb,amsfonts}
\usepackage{tabularx,booktabs}
\usepackage{mathtools}
\usepackage{algorithmic}
\usepackage{graphicx}
\usepackage{textcomp}
\usepackage{xcolor}
\usepackage{placeins}
\usepackage{dsfont}
\usepackage{hyperref}
\usepackage{nccmath}
\def\BibTeX{{\rm B\kern-.05em{\sc i\kern-.025em b}\kern-.08em
    T\kern-.1667em\lower.7ex\hbox{E}\kern-.125emX}}
\newcolumntype{C}{>{\centering\arraybackslash}X} % centered version of "X" type
\newcolumntype{b}{>{\hsize=2.3\hsize}X}
\usepackage{amsthm}
\usepackage{bbm}
\usepackage{csquotes}
\usepackage{chngcntr}
\usepackage{apptools}
\usepackage{float}
\usepackage{cuted, ragged2e}

\usepackage{cite}

\theoremstyle{plain}
\newtheorem{theorem}{Theorem}
\newtheorem*{theorem*}{Theorem}
\newtheorem*{corollary*}{Corollary}
\newtheorem{lemma}{Lemma}
\newtheorem{corollary}{Corollary}
\newtheorem{proposition}{Proposition}

\theoremstyle{definition}
\newtheorem{definition}{Definition}

\theoremstyle{remark}

\newtheorem{remark}{Remark}
\newtheorem{example}{Example}
\newtheorem{counterexample}{Counterexample}
\DeclareMathOperator*{\argmax}{arg\,max}

\usepackage{amsmath}
\DeclareMathOperator{\sign}{sign}

\newcommand{\X}{\mathcal{X}}
\newcommand{\Y}{\mathcal{Y}}

\newcommand{\Pm}{\mathcal{P}}
\newcommand{\Q}{\mathcal{Q}}
\newcommand{\F}{\mathcal{F}}

\makeatletter
\newcounter{labelcnt}
\renewcommand{\thelabelcnt}{(\alph{labelcnt})}
\newcommand{\setlabel}[1]{%
  \refstepcounter{labelcnt}\ltx@label{lbl:#1}%
  {\text{\upshape\thelabelcnt}}%
}

\makeatother
\DeclareMathOperator*{\esssup}{ess\,sup}
\DeclareMathOperator*{\essinf}{ess\,inf}
\newcommand{\ml}[2]{\mathcal{L}\left(#1  \!\!  \to  \!\!   #2\right)} % maximal leakage

\begin{document}

\title{On Sibson's $\alpha$-Mutual Information}

\author{
\IEEEauthorblockN{Amedeo Roberto Esposito, Adrien Vandenbroucque, Michael Gastpar} 
\IEEEauthorblockA{\textit{School of Computer and Communication Sciences} \\
       EPFL, Lausanne, Switzerland\\
\{amedeo.esposito, michael.gastpar\}@epfl.ch, adrien.vandenbroucque@alumni.epfl.ch}
}

\maketitle

\begin{abstract} 
   We explore a family of information measures that stems from R\'enyi's $\alpha$-Divergences with $\alpha<0$. In particular, we extend the definition of Sibson's $\alpha$-Mutual Information to negative values of $\alpha$ and show several properties of these objects. Moreover, we highlight how this family of information measures is related to functional inequalities that can be employed in a variety of fields, including lower-bounds on the Risk in Bayesian Estimation Procedures.
\end{abstract}

\begin{IEEEkeywords}
R\'enyi-Divergence, Sibson's Mutual Information, Information Measures, Bayesian Risk, Estimation
\end{IEEEkeywords}

\section{Introduction}
Sibson's $\alpha$-Mutual Information is a generalization of Shannon's Mutual Information with several applications in probability, information and learning theory \cite{fullVersionGeneralization}. In particular, it has been used to provide concentration inequalities in settings where the random variables are \textbf{not} independent, with applications to learning theory \cite{fullVersionGeneralization}. The measure is also connected to Gallager's exponent function, a central object in the channel coding problem both for rates below and above capacity \cite{gallager1,gallager2}. Moreover, a new operational meaning has been given to the measure with $\alpha\!=\!+\infty$ when a novel measure of information leakage has been proposed in \cite{leakageLong}, under the name of Maximal leakage. %Maximal Leakage has been defined with a stark operational meaning and it has been shown to be equal to Sibson's $\infty$-Mutual Information. 
Similarly to $I_\alpha$, Maximal Leakage has recently found applications in learning and probability theory \cite{fullVersionGeneralization}. In this work we will look at some properties of $I_\alpha$ with a different perspective and we will extend the definition in order to include negative values of $\alpha$ as well. The reason for this extension is strongly operationally driven and will be tied to a family of functional-analytic bounds that we can provide for Sibson's $\alpha$ MI for every $\alpha\in\mathbb{R}$.
\section{Background and definitions}
Throughout the paper we will often use the notion of $L^p$-norms: let $1\leq p \leq \infty$ and consider the measurable space $(\Omega, \mathcal{F}, \mu)$, let $f$ be a measurable function with respect to the space, then
\begin{equation}
    \left\lVert f \right\rVert_{L^p(\mu)} = \left(\int |f|^pd\mu\right)^\frac1p.
\end{equation}
The definition can be extended to values of $p<1$, however in those cases one does not have norms anymore e.g., if $0<p<1$ one recovers a quasi-norm but not a norm. 
\subsection{Sibson's $\alpha$-Mutual Information}
Introduced by R\'enyi as a generalization of entropy and KL-divergence, $\alpha$-divergence has found many applications ranging from hypothesis testing to guessing and several other statistical inference and coding problems~\cite{verduAlpha}. Indeed, it has several useful operational interpretations (e.g., hypothesis testing, and the cut-off rate in block coding \cite{RenyiKLDiv,opMeanRDiv1}). It can be defined as follows~\cite{RenyiKLDiv}:
\begin{definition}
	Let $(\Omega,\F,\Pm),(\Omega,\F,\Q)$ be two probability spaces. Let $\alpha>0$ be a positive real number different from $1$. Consider a measure $\mu$ such that $\Pm\ll\mu$ and $\Q\ll\mu$ (such a measure always exists, e.g. $\mu=(\Pm+\Q)/2$) and denote with $p,q$ the densities of $\Pm,\Q$ with respect to $\mu$. The $\alpha$-Divergence of $\Pm$ from $\Q$ is defined as follows:
	\begin{align}
	D_\alpha(\Pm\|\Q)=\frac{1}{\alpha-1} \log \int p^\alpha q^{1-\alpha} d\mu.
	\end{align}
\end{definition}
\begin{remark}
    The definition is independent of the chosen measure $\mu$. %whenever $\infty>\alpha>0$ and $\alpha\neq 1$. 
    It is indeed possible to show that
    $\int p^{\alpha}q^{1-\alpha} d\mu = \int \left(\frac{q}{p}\right)^{1-\alpha}d\Pm $, and that whenever $\Pm\ll\Q$ or $0<\alpha<1,$ we have $\int p^{\alpha}q^{1-\alpha} d\mu= \int \left(\frac{p}{q}\right)^{\alpha}d\Q$, see \cite{RenyiKLDiv}.
\end{remark}

It can be shown that if $\alpha>1$ and $\Pm\not\ll\Q$ then $D_\alpha(\Pm\|\Q)=\infty$. The behaviour of the measure for $\alpha\in\{0,1,\infty\}$ can be defined by continuity. In general, one has that $D_1(\Pm\|\Q) = D(\Pm\|\Q)$ but if $D(\Pm\|\Q)=\infty$ or there exists $\beta$ such that $D_\beta(\Pm\|\Q)<\infty$ then $\lim_{\alpha\downarrow1}D_\alpha(\Pm\|Q)=D(\Pm\|\Q)$\cite[Theorem 5]{RenyiKLDiv}. For an extensive treatment of $\alpha$-divergences and their properties we refer the reader to~\cite{RenyiKLDiv}. 
Starting from R\'enyi's Divergence and the geometric averaging that it involves, Sibson built the notion of Information Radius \cite{infoRadius}:
\begin{definition}\label{SibsonsInfoRadius}
    Let $(\mu_1,\ldots,\mu_n)$ be a family of probability measures and $(w_1,\ldots, w_n)$ be a set of weights s.t. $w_i\geq 0$ for $i=1,\ldots,n$ and such that $\sum_{i=1}^n w_i>0$. Let $\alpha\geq1$, the information radius of order $\alpha$ is defined as:
    \begin{align*}
        \frac{1}{\alpha-1}\min_{\nu\ll\sum_iw_i\mu_i}\log\left(\sum_i w_i\exp((\alpha-1)D_\alpha(\mu_i\|\nu))\right)  \label{infoRadius}.
    \end{align*}
\end{definition}

Suppose now we have two random variables $X,Y$ jointly distributed according to $\Pm_{XY}$. It is possible to generalise Def. \ref{SibsonsInfoRadius} and see that the information radius is a special case of the following quantity \cite{verduAlpha}:
\begin{equation}
    I_\alpha(X,Y) = \min_{\Q_Y} D_\alpha(\Pm_{XY}\|\Pm_{X}\Q_Y). \label{sibsIAlpha} 
\end{equation}
$I_\alpha(X,Y)$ represents a generalisation of Shannon's Mutual Information and possesses many interesting properties \cite{verduAlpha}. Indeed, $\lim_{\alpha\to 1}I_\alpha(X,Y)=I(X;Y)$. 
On the other hand when $\alpha\to\infty$, we get: $I_\infty(X,Y)=\log\mathbb{E}_{\Pm_Y}\left[\sup_{x:\Pm_X(x)>0} \frac{\Pm_{XY}(\{x,Y\})}{\Pm_X(\{x\})\Pm_Y(\{Y\})}\right] =\ml{X}{Y}$, where $\ml{X}{Y}$ denotes the Maximal Leakage from $X$ to $Y$, a recently defined information measure with an operational meaning in the context of privacy and security \cite{leakageLong}.
For more details on Sibson's $\alpha$-MI, as well as a closed-form expression, we refer the reader to~\cite{verduAlpha}, as for Maximal Leakage the reader is referred to~\cite{leakageLong}.
\subsection{R\'enyi's $\alpha$-Divergence - Negative Orders}
Not much is known or has been explored on R\'enyi's $\alpha$-Divergence with $\alpha<0$. According to R\'enyi himself \cite{renyiMeasuresInfo} only positive orders can be regarded as information measures. One of the reasons behind this statement is probably the fact that, taking an axiomatic approach to measures of information like the one undertaken in \cite{renyiMeasuresInfo}, $D_\alpha$ with $\alpha<0$ satisfies many of those properties with the wrong sign of inequality. In fact\cite{RenyiKLDiv}:
\begin{proposition}
Let $\alpha<0$ and $\mu,\nu$ be two probability measures such that $\nu\ll\mu$ then:
\begin{enumerate}
    \item $D_\alpha(\nu\|\mu) \leq 0$ (as opposed to \textbf{$\geq 0$});
    \item $D_\alpha(\nu\|\mu)$ is \textbf{concave} in $\nu$ (as opposed to \textbf{convex} in $\mu$);
    \item $D_\alpha$ is \textbf{upper} semi-continuous in the pair $(\nu,\mu)$ in the topology of set-wise convergence (as opposed to \textbf{lower} semi-continuous) ;
    \item Let $K$ be a Markov Kernel, then $D_\alpha(K\nu\|K\mu) \geq D_\alpha(\nu\|\mu)$ (as opposed to $D_\alpha(K\nu\|K\mu)\leq  D_\alpha(\nu\|\mu)$);
\end{enumerate}
\end{proposition}

However, some of the properties are maintained, like the following:
\begin{theorem*}[\!\!{\cite[Thm 39]{RenyiKLDiv}}]
Let $\alpha\in[-\infty,+\infty]$ and let $\nu,\mu$ be two probability measures such that $\nu\ll\mu$, then $D_\alpha(\nu\|\mu)$ is non-decreasing in $\alpha$.
\end{theorem*}
And to conclude, let us present a result connecting negative and positive orders:
\begin{lemma}[\!\!{{\cite[Lemma 10]{RenyiKLDiv}}} - Skew Symmetry]
For every $\alpha\in(-\infty,+\infty)\setminus\{0,1\}$, let let $\nu,\mu$ be two probability measures such that $\nu\equiv\mu$
\begin{equation}
    D_\alpha(\nu\|\mu) = \frac{\alpha}{1-\alpha}D_{(1-\alpha)}(\mu\|\nu). 
\end{equation}
\end{lemma}
Such a relationship, although very useful, is restricted to cases in which the two measures are equivalent with respect to each other, \textit{i.e.}, absolute continuity holds in both directions: $\nu\ll\mu$ \textbf{and} $\mu\ll\nu$ otherwise one might incur in settings where one divergence is finite and the other is infinite. However, $D_\alpha$ with $\alpha<0$ is not concave in the second argument. 
\begin{counterexample}
Consider a discrete setting and point mass functions. Let $\mu_1=(0.32,0.68),\mu_2=(0.5,0.5),\nu=(0.13,0.87)$ and $\lambda=0.4$. One has that with $\alpha=-2$, $D_\alpha(\nu\|\mu_1)=-0.2855$, $D_\alpha(\nu\|\mu_2)=-0.6744$ and, moreover, $-0.5287=D_\alpha(\nu\|\lambda\mu_1+(1-\lambda)\mu_2)< \lambda D_\alpha(\nu\|\mu_1)+(1-\lambda) D_\alpha(\nu\|\mu_2) =-0.5188.$
\end{counterexample}
\section{Definition}
The starting point will, once again, be \eqref{sibsIAlpha}. From now on we will consider $\alpha$ to be always strictly negative unless specified otherwise and we will restrict ourselves to discrete probability measures for simplicity. \iffalse Given the non-positivity and potential unboundedenss of $D_\alpha$ we define a Sibson's $\alpha$-MI in this case through a minimisation, like in \eqref{sibsIAlpha}, would make little sense, since one can always select a measure $\mathcal{Q}_Y$ such that $\Pm_{XY}\not\ll \Pm_X\mathcal{Q}_Y$ and consequently $I_\alpha(X,Y)$ would trivially be $-\infty$ (Excluding non-absolutely continuous measures leads to what?). Moreover, such a minimisation problem can be solved because with $\alpha>0$, $D_\alpha(\mu\|\nu)$ is indeed convex in $\nu$. As of now we do not have concavity of $D_\alpha(\mu\|\nu)$ with respect to $\nu$ when $\alpha<0$, so how can we justify the maximisation problem? Can we prove concavity?\fi
We thus bring forth the following definition:
\begin{definition}
Let $X,Y$ be two random variables whose joint measure is $\Pm_{XY}$ and the corresponding marginals are given by $\Pm_X$ and $\Pm_Y$. Let $\alpha\in(-\infty,0)$
\begin{equation}
    I_\alpha(X,Y) = - \max_{\Q_Y} D_\alpha(\Pm_{XY}\|\Pm_X\Q_Y). \label{definition}
\end{equation}

\end{definition}
\begin{remark}
The minus in the definition is included in order to make sure that all the properties of an information measures are satisfied (non-negativity, DPI, etc.). A similar approach can be undertaken in order to define $D_\alpha$ and $H_\alpha$ with negative $\alpha$ by simply changing the multiplicative constant from $\frac{1}{\alpha-1}$ to $\frac{1}{1-\alpha}$ for R\'enyi's divergence and from $\frac{1}{1-\alpha}$ to $\frac{1}{\alpha-1}$ for R\'enyi's entropy for negative orders.
\end{remark}
This quantity has, much like Sibson's $\alpha$-MI with $\alpha>0$, a closed-form expression given by the following result.
\begin{theorem}
Let $X,Y$ be two random variables whose joint measure is $\Pm_{XY}$ and the corresponding marginals are given by $\Pm_X$ and $\Pm_Y$. Let $\alpha<0$
\begin{equation}
    I_\alpha(X,Y) = -\frac{\alpha}{
    \alpha-1}\log\mathbb{E}_{\Pm_Y}\left[\mathbb{E}_{\Pm_X}^{\frac1\alpha}\left[\left(\frac{d\Pm_{XY}}{d\Pm_X\Pm_Y}\right)^\alpha\right]\right].
\end{equation}
\end{theorem}
\begin{proof}
We have that for every $\Q_Y$, \begin{align}D_\alpha(\Pm_{XY}\|\Pm_X\Q_Y)&=\frac{1}{\alpha-1}\log\mathbb{E}_{\Q_Y}\left[\mathbb{E}_{\Pm_X}\left[\Pm_{Y|X}^\alpha \Q_Y^{-\alpha}\right]\right] \\ &=
\frac{\alpha}{\alpha-1}\log\mathbb{E}^\frac1\alpha_{\Q_Y}\left[\frac{\mathbb{E}_{\Pm_X}\left[\Pm_{Y|X}^\alpha\right]}{Q_Y^\alpha}\right] \\
&\leq \frac{\alpha}{\alpha-1}\log\mathbb{E}_{\Q_Y}\left[\frac{\mathbb{E}^{\frac1\alpha}_{\Pm_X}\left[\Pm_{Y|X}^\alpha\right]}{Q_Y}\right] \label{Jensen}\\
&=  \frac{\alpha}{\alpha-1}\log\sum_y\mathbb{E}^{\frac1\alpha}_{\Pm_X}\left[\Pm_{Y|X}^\alpha\right]\\
&=  \frac{\alpha}{\alpha-1}\log\sum_y\mathbb{E}^{\frac1\alpha}_{\Pm_X}\left[\Pm_{XY}^\alpha\Pm_X^{-\alpha}\right]\\
&=  \frac{\alpha}{\alpha-1}\log\mathbb{E}_{\Pm_Y}\left[\mathbb{E}^{\frac1\alpha}_{\Pm_X}\left[\left(\frac{\Pm_{X|Y}}{\Pm_X}\right)^\alpha\right]\right] \\
&= -I_\alpha(X,Y).
\end{align}
Where \eqref{Jensen} follows from the convexity of $x^{\frac1k}, k<0$ and Jensen's inequality.
This implies that for every $\Q_Y$:
\begin{align}
   -I_\alpha(X,Y) \geq D_\alpha(\Pm_{XY}\|\Pm_X\Q_Y). \label{IalphaClosedFormExpression}
\end{align}
Moreover one has that, given
\begin{equation}
    \Q_Y^\star(y) = \frac{\Pm_Y(y)\left(\sum_{x}\Pm_{X|Y=y}(x)^\alpha \Pm_X(x)^{1-\alpha}\right)^\frac1\alpha}{\mathbb{E}_{\Pm_Y}\left[\left(\sum_{x}\Pm_{X|Y=y}(x)^\alpha \Pm_X(x)^{1-\alpha}\right)^\frac1\alpha\right]}
\end{equation}
then,
\begin{equation}
    D_\alpha(\Pm_{XY}\|\Pm_X\Q^\star_Y) = - I_\alpha(X,Y).
\end{equation}
\end{proof}
\begin{remark}
In a setting where both $X$ and $Y$ are discrete random variables then denoting with $p_{XY}$, $p_Y$ and $p_X$ the joint point mass functions  and the corresponding marginals and denoting with $p_{X|Y=y}$ the conditional point mass functions one has that:
\begin{align}
I_\alpha(X,Y)\! &=\! \frac{\alpha}{
    1-\alpha}\!\log\!\sum_y\! p_Y(y)\!\!\left(\!\sum_x\! \left(\frac{p_{X|Y=y}(x)}{p_{X}(x)}\!\right)^\alpha\!\!\! p_{X}(x)\!\!\right)^{\!\!\frac{1}{\alpha}} \notag \\
    &= \frac{\alpha}{
    1-\alpha}\!\log\!\sum_y \!\!\left(\!\sum_x p_{Y|X=x}(y)^\alpha p_{X}(x)\right)^{\!\!\frac{1}{\alpha}}.
\end{align}
\end{remark}
\begin{theorem}
Let $X,Y$ be two discrete random variables whose joint point mass function is $p_{XY}$ and the corresponding marginals are given by $p_X$ and $p_Y$ then
\begin{equation}
    I_{-\infty}(X,Y) = - \log\sum_y\left(\min_{x:p_X(x)>0} \Pm_{Y|X=x}(y)\right).
\end{equation}
\begin{proof}
The result follows from noticing that one can re-write $I_\alpha$ as follows
\begin{equation}
    -I_\alpha(X,Y)= \frac{\alpha}{\alpha-1}\log\mathbb{E}_{\Pm_Y}\left[\left\lVert\frac{\Pm_{X|Y}}{\Pm_X}\right\rVert_{L^\alpha(\Pm_X)}\right],
\end{equation}
where $\left\lVert\frac{\Pm_{X|Y}}{\Pm_X}\right\rVert_{L^\alpha(\Pm_X)}$ denotes the ''$\alpha-$norm'' with respect to $\Pm_X$. %\MG{Is this clear enough? --- I still feel that there would be a benefit from talking about norms early on, like in Section II.--- well, I suppose we could say that we only use this notation twice, once here and once below in equation (30), so maybe it's fine to be approximate about it.} 
Taking the limit of $\alpha\to-\infty$ one has that $\frac{\alpha}{\alpha-1}\to 1$ and $\left\lVert\frac{\Pm_{X|Y}}{\Pm_X}\right\rVert_{L^\alpha(\Pm_X)}\to \essinf_{\Pm_X} \frac{\Pm_{X|Y}}{\Pm_X}.$
The conclusion then follows from simple algebraic manipulations.
\end{proof}
\end{theorem}
\begin{remark}
The quantity $I_{-\infty}$ has already appeared in the literature as ``Maximal-Cost Leakage'' \cite[Thm. 15, Eq. (95)]{leakageLong}. Indeed, assuming equivalence between $\Pm_{XY}$ and $\Pm_X\Q_Y$ and using the Skew Symmetry of $D_\alpha$ one has that $-D_{-\infty}(\Pm_{XY}\|\Pm_X\Q_Y)=D_\infty(\Pm_X\Q_Y\|\Pm_{XY})$ and consequently 
\begin{align}
    I_{-\infty}(X,Y)&=-\max_{\Q_Y} D_{-\infty}(\Pm_{XY}\|\Pm_X\Q_Y) \\ &= \min_{\Q_Y}D_\infty(\Pm_X\Q_Y\|\Pm_{XY}) \\
    &= \mathcal{L}^{c}(X\!\!\to\!\!Y).
\end{align}
\end{remark}
\section{Properties}
For $I_\alpha$ with $\alpha<0$ we can show the following properties.
\begin{theorem}
Let $X,Y$ be two random variables such that $\Pm_{XY}\ll\Pm_X\Pm_Y$ and assume that $\alpha<0$, then:
\begin{enumerate}
    \item $I_\alpha(X,Y)\geq 0$ with equality iff $X$ and $Y$ are independent; \label{nonNegativity}
    \item $I_\alpha(X,Y) \neq I_\alpha(Y,X)$; \label{asymmetry}
    \item Let $0>\alpha_1>\alpha_2$ then $I_{\alpha_1}(X,Y)\leq I_{\alpha_2}(X,Y)$; \label{nonIncreasability}
    \item $I_\alpha(X,Y)\leq \mathcal{L}^c(X\!\!\to\!\!Y)$ for every $-\infty<\alpha<0$; \label{mcLeakageProperty}
    \item Let $X-Y-Z$ be a Markov Chain, $I_\alpha(X,Z)\leq \min\{ I_\alpha(X,Y), I_\alpha(Y,Z) \}$; \label{DPI}
    \item $\exp\left(-\frac{\alpha-1}{\alpha}I_\alpha(X,Y)\right)$ is concave in $\Pm_{Y|X}$ and $I_\alpha(X,Y)$ is convex in $\Pm_{Y|X}$. \label{convexity} 
\end{enumerate}
\end{theorem}
\begin{proof}

\ref{nonNegativity}): We have that $I_\alpha(X,Y)=-D_{\alpha}(\Pm_{XY}\|\Pm_X\Q_Y^\star)$ and given that $D_\alpha\leq 0$ if $\alpha\leq 0$ we have the non-negativity of $I_\alpha$. Moreover, if $X$ and $Y$ are independent then $\Q^\star_Y(y)= \Pm_Y(y)$ and $D_\alpha(\Pm_{XY}\|\Pm_X\Pm_Y)=0$. On the other hand, if $I_\alpha(X,Y)=0$ then $D_\alpha(\Pm_{XY}\|\Pm_X\Q^\star_Y)=0$ which means that $\Pm_{Y|X=x}=\Q_Y^\star$, hence $Y$ does not depend on $X$.
%then by Property \ref{mcLeakageProperty}) we have that $\mathcal{L}^c(X\!\!\to\!\!Y)=0$ which implies independence by \cite[Corollary 5.2)]{leakageLong}.
\\
\ref{asymmetry}): It's clear from the definition and shown more concretely in Example \ref{exampleAsymm}.\\
\ref{nonIncreasability}): We have that, denoting with $\Q_Y^{\alpha_2}=\argmax_{\Q_Y} D_{\alpha_2}(\Pm_{XY}\|\Pm_X\Q_Y)$
\begin{align}
    -I_{\alpha_1}(X,Y) & = \max_{Q_Y}D_{\alpha_1}(\Pm_{XY}\|\Pm_X\Q_Y) \\
    &\geq D_{\alpha_1}(\Pm_{XY}\|\Pm_X\Q_Y^{\alpha_2}) \\
    &\geq D_{\alpha_2}(\Pm_{XY}\|\Pm_X\Q_Y^{\alpha_2}) \label{nonDecreas}
    \\
    &= -I_{\alpha_2}(X,Y)
\end{align}
where \eqref{nonDecreas} follows from the non-decreasability of $D_\alpha$ for $\alpha\in[-\infty,\infty]$ \cite[Thm. 39]{RenyiKLDiv}.
\\
\ref{mcLeakageProperty}): Follows from \ref{nonIncreasability}) and the fact that $\mathcal{L}^c(X\!\!\to\!\!Y)=$ $\lim_{\alpha\to-\infty}I_\alpha(X,Y)$.\\
\ref{DPI}): let $X-Y-Z$ be a Markov chain and let $K_{Z|Y}$ denote the corresponding Markov Kernel $K_{Z|Y}:(\mathcal{Y},\mathcal{F}) \to (\mathcal{Z},\mathcal{\hat{F}})$ induced by the transition probabilities $\Pm_{Z|Y}$. Consequently, given any probability measure $\Q_Y$ one can construct a corresponding $\Q_Z=\Q_Y K_{Z|Y}$ where substantially $\Q_Z(z)=\mathbb{E}_{\Q_Y}[\Pm_{Z|Y}(z)]$. By the Data-Processing Inequality for $D_\alpha$ one has that for every $\Q_Y$:
\begin{equation}
    -D_\alpha(\Pm_{XZ}\|\Pm_X\Q_Z) \leq -D_\alpha(\Pm_{XY}\|\Pm_X\Q_Y)
\end{equation}
Taking the $\inf$ with respect to $\Q_Y$ leads to 
\begin{equation}
    -D_\alpha(\Pm_{XZ}\|\Pm_X\Q^\star_Z) \leq I_\alpha(X,Y)
\end{equation}
and given that $I_\alpha(X,Z)\leq -D_\alpha(\Pm_{XZ}\|\Pm_X\Q_Z)$ for every $\Q_Z$ the statement follows. 
The other inequality follows by observing that considering the Kernel determined by $\Pm_{XZ|XY}$ and denoted by $K_{XZ|YZ}$ one has that $\Pm_{XZ}=\Pm_{YZ}K_{XZ|YZ}$ while $\Pm_X\Q_Z = (\Pm_X\Q_Y)K_{XZ|YZ}$, consequently by the Data-Processing Inequality for $D_\alpha$ one has that 
\begin{equation}
    -D_\alpha(\Pm_{XZ}\|\Pm_X\Q_Z) \leq -D_\alpha(\Pm_{YZ}\|\Pm_Y\Q_Z).
\end{equation}
The statement follows from a similar argument as above.
\\
\ref{convexity}): One can rewrite the expression as follows \begin{equation}\exp\left(-\frac{\alpha-1}{\alpha}I_\alpha(X,Y)\right)= \sum_y \left\lVert \Pm_{Y|X} \right\rVert_{L^\alpha(\Pm_X)}\end{equation}
and concavity follows from the Reverse-Minkowski's inequality.
Convexity of $I_\alpha(X,Y)$ follows from the concavity just proven and the fact that $-\frac{\alpha}{\alpha-1}\log(x)$ is a non-increasing convex function for a given $\alpha$ and composition of a non-increasing convex function with a concave one gives rise to a convex function \cite[Eq (3.10), pag. 84]{boyd}.
\end{proof}
Following \cite{verduAlpha} we will now look at some specific choices of $X$ and $Y$ and compute the corresponding values of $I_\alpha$. Given that the expression for the information measure is essentially identical for $\alpha<0$ and $\alpha$, the values in these examples match the one for $I_\alpha$ with $\alpha>0$ except for a minus sign that ensures the non-negativity of the measure.
\begin{example}
Let $X,Y\sim\text{Ber}(1/2)$ and let $\mathbb{P}(Y\neq X)=\delta$ then if $\alpha<0$
\begin{align}
    I_{\alpha}(X,Y)= I_\alpha(Y,X) = - d_\alpha(\delta\|1/2),
\end{align}
where $d_\alpha(p\|q)=\frac{1}{\alpha-1}\log(p^\alpha q^{1-\alpha}+(1\!-p)^\alpha(1\!-q)^{1-\alpha})$ denotes the binary $\alpha$-divergence.
\end{example}
\begin{example}\label{exampleAsymm}
Let $X\sim\text{Ber}(1/2)$ and let $Y\in\{0,1,e\}$. Assume also that 
\begin{equation}
   P_{Y|X=x}(y)=
    \begin{cases}
        1-\delta,\,\, x=y \\
        \delta,\,\, y=e \\
        0,\text{ else}.
    \end{cases}
\end{equation}
Then, if $\alpha<0$ 
\begin{align}
    I_\alpha(X,Y) &= -\frac{\alpha}{\alpha-1}\log_2\left(\delta+(1-\delta)2^\frac{\alpha-1}{\alpha}\right)  \\
    I_\alpha(Y,X) &=  -\frac{1}{\alpha-1}\log_2\left(\delta+(1-\delta)2^{\alpha-1}\right)
\end{align}
\end{example}
\begin{example}
Let $X\sim \mathcal{N}(0,\sigma^2_X)$ and $N\sim \mathcal{N}(0,\sigma^2_N)$ with $N$ independent from $X$ then if $-\frac{\sigma^2_N}{\sigma^2_X}<\alpha<0$
\begin{align}
    I_\alpha(X,X+N) &= -\frac{1}{2}\log\left(1+\alpha\frac{\sigma^2_X}{\sigma^2_N}\right).
\end{align}
This example shows a major difference in behaviour between the two families of information measures: the value of $I_\alpha(X,X+N)$ has the familiar expression only if $|\alpha|\leq \frac{\sigma^2_N}{\sigma^2_X}$.% The computations clearly show that in order to ``complete the square'' and achieve this expression one needs to restrict $\alpha$. 
\end{example}
\section{Bound}
\begin{theorem}\label{lowerBoundHolder}
Let $X,Y$ be two random variables whose joint measure is $\Pm_{XY}$ and the corresponding marginals are given by $\Pm_X$ and $\Pm_Y$.
For every $f:\mathcal{X}\times\mathcal{Y}\to\mathbb{R}^+$ one has that 
\begin{align}
    \mathbb{E}_{\Pm_{XY}}[f(X,Y)] \geq &\mathbb{E}^{\frac{1}{\beta'}}_{\Pm_Y}\left[\mathbb{E}^{\frac{\beta'}{\beta}}_{\Pm_X}\left[f(X,Y)^\beta\right]\right]\\ &\cdot\mathbb{E}^{\frac{1}{\alpha'}}_{\Pm_Y}\left[\mathbb{E}^{\frac{\alpha'}{\alpha}}_{\Pm_X}\left[\left(\frac{d\Pm_{XY}}{d\Pm_X\Pm_{Y}}\right)^{\alpha}\right]\right],
\end{align}
where $\frac1\alpha+\frac1\beta = 1 = \frac{1}{\alpha'}+\frac{1}{\beta'}$ and $\alpha,\alpha'<1$. 
\end{theorem}
\begin{proof}
\begin{align}
 \mathbb{E}_{\Pm_{XY}}[f(X,Y)] =\, &\mathbb{E}_{\Pm_X\Pm_Y}\left[f(X,Y)\left(\frac{d\Pm_{XY}}{d\Pm_X\Pm_{Y}}\right)\right] \\ 
 \geq\, &\mathbb{E}_{\Pm_Y}\Bigg[\mathbb{E}^{\frac{1}{\beta}}_{\Pm_X}\left[f(X,Y)^\beta\right]\\\, &\,\,\,\,\,\,\,\,\,\,\,\,\,\, \cdot\mathbb{E}^{\frac{1}{\alpha}}_{\Pm_X}\left[\left(\frac{d\Pm_{XY}}{d\Pm_X\Pm_{Y}}\right)^{\alpha}\right]\Bigg] \\ 
 \geq\, &\mathbb{E}^{\frac{1}{\beta'}}_{\Pm_Y}\left[\mathbb{E}^{\frac{\beta'}{\beta}}_{\Pm_X}\left[f(X,Y)^\beta\right]\right]\\ \cdot&\,\mathbb{E}^{\frac{1}{\alpha'}}_{\Pm_Y}\left[\mathbb{E}^{\frac{\alpha'}{\alpha}}_{\Pm_X}\left[\left(\frac{d\Pm_{XY}}{d\Pm_X\Pm_{Y}}\right)^{\alpha}\right]\right],
\end{align} 
where each inequality follows from applying the reverse H\"older's inequality (which, in turn, requires the positivity of $f$).
\end{proof}
Taking the limit of $\alpha'\to 1$ (which implies $\beta'\to-\infty$) leads to the following bound involving $I_\alpha(X,Y)$ with $\alpha<1$:
\begin{corollary}\label{sibsIalphaBoundLess1}
Consider the same setting as in Theorem \ref{lowerBoundHolder}. For every $f:\mathcal{X}\times\mathcal{Y}\to\mathbb{R}^+$ and for every $\alpha<1$ and $\beta=\frac{\alpha-1}{\alpha}$ one has that
\begin{align}
 \mathbb{E}_{\Pm_{XY}}[f(X,Y)] \geq  &\essinf_{\Pm_Y}\left[\mathbb{E}^{\frac{1}{\beta}}_{\Pm_X}\left[f(X,Y)^\beta\right]\right]\\ &\cdot\exp\left(\sign(\alpha)\cdot\frac{\alpha-1}{\alpha}I_\alpha(X,Y)\right).
\end{align} 
\end{corollary}
Corollary \ref{sibsIalphaBoundLess1} holds for every non-negative function $f$. However, given $0<\alpha<1$ (which in turn implies $\beta<0$), if there exists an $x$ with positive measure with respect to $\Pm_X$ and such that for every $y$ such that $\Pm_Y(\{y\})>0$, $f(x,y)=0$, one recovers a trivial lower-bound on $\mathbb{E}_{\Pm_{XY}}[f(X,Y)]$. This prevents us from setting $f = \mathds{1}_E$ and recovering a bound that involves probabilities in this case. However, whenever $\alpha<0$ (which implies $0<\beta<1$) we do recover said functions and we can provide the following lower-bound connecting the probability of any event $E$ with respect to the joint $\Pm_{XY}$ and the product of the marginals $\Pm_X\Pm_Y$:
\begin{corollary}\label{sibsIalphaBoundLess1Prob}
Let $X,Y$ be two random variables and consider the probability spaces $(\X\times\Y,\mathcal{F},\Pm_{XY})$ and $(\X\times\Y,\mathcal{F},\Pm_X\Pm_Y)$. Let $E\in\mathcal{F}$ and, given $y\in\Y$, denote with $E_y = \{x:(x,y)\in E\}$, then, for every $\alpha<0$
\begin{align}
    \Pm_{XY}(E) &\geq \min_{y}\Pm_X(E_y)^{\frac1\beta}\cdot \exp\left(-\frac{\alpha-1}{\alpha}I_\alpha(X,Y)\right) \\
    &= \exp\left(\frac{1}{\beta}\left(\log(\min_y\Pm_X(E_y))-I_\alpha(X,Y)\right)\right).
\end{align}
Taking the limit of $\alpha\to-\infty$ which implies $\beta\to 1$ one recovers the following:
\begin{align}
    \Pm_{XY}(E) &\geq \min_{y}\Pm_X(E_y)\exp(-I_{-\infty}(X,Y)) \\ &= \min_{y}\Pm_X(E_y)\exp(-\mathcal{L}^c(X\!\!\to\!\!Y)).
\end{align}
\end{corollary}
\begin{table*}[!hbpt]

\caption{Behaviour of the bounds expressed in Corollary \ref{sibsIalphaBoundLess1}, Corollary \ref{sibsIalphaBoundLess1Prob}, Corollary \ref{sibsPosAlphaFunctions} and  \cite[Corollary 1]{fullVersionGeneralization}}
\label{comparison}
\centering
\begin{tabular}{ |p{3.1cm}||p{4cm}|p{4cm}|p{4cm}| } 
 \hline \multicolumn{4}{|c|}{ } \\[-7pt]
 \multicolumn{4}{|c|}{Behaviour of the Bound $\mathbb{E}_{\Pm_{XY}}[f(X,Y)]	\lesseqgtr h_\beta(f(X,Y))\cdot g(I_\alpha(X,Y))$} \\[4pt]
 \hline
  Admitted values of $\alpha$ & $\alpha<0 \implies 0<\beta<1$ &$0<\alpha<1 \implies \beta<0$ &$\alpha>1 \implies \beta>1$\\
 \hline
 Information-Measure\! $g(I_\alpha)$\!& $\exp((1-\alpha)/\alpha\cdot I_\alpha(X,Y))$   &$\exp((\alpha-1)/\alpha\cdot I_\alpha(X,Y))$&   $\exp((\alpha-1)/\alpha\cdot I_\alpha(X,Y))$\\
 Multiplicative Term $h_\beta(f)$ &$\min_{y}\mathbb{E}^\frac1\beta_{P_X}\left[f(X,y)^{\beta}\right]$
 &$\min_{y}\mathbb{E}^\frac1\beta_{P_X}\left[f(X,y)^{\beta}\right]$ &$\max_{y}\mathbb{E}^\frac1\beta_{P_X}\left[f(X,y)^{\beta}\right]$  \\
 Multiplicat. Term $h_\beta(\mathds{1}_E)$ &   $\min_y(P_X(E_y))^\frac1\beta$  & cannot be provided  & $\max_y (P_X(E_y))^\frac1\beta$\\  Inequality & $\mathbb{E}_{\Pm_{XY}}[f]\geq h_\beta(f)\cdot g(I_\alpha(X,Y))$ & $\mathbb{E}_{\Pm_{XY}}[f]\geq h_\beta(f)\cdot g(I_\alpha(X,Y))$ & $\mathbb{E}_{\Pm_{XY}}[f]\leq h_\beta(f)\cdot g(I_\alpha(X,Y))$ \\
 References & Corollary \ref{sibsIalphaBoundLess1} and Corollary \ref{sibsIalphaBoundLess1Prob} & Corollary \ref{sibsIalphaBoundLess1}  &Corollary \ref{sibsPosAlphaFunctions} and \cite[Corollary 1]{fullVersionGeneralization} \\
 \hline
\end{tabular}
\end{table*}
Let us compare, through Table \ref{comparison}, Corollary \ref{sibsIalphaBoundLess1}, Corollary \ref{sibsIalphaBoundLess1Prob} and a straight-forward generalisation of \cite[Corollary 1]{fullVersionGeneralization} that we will now state for reference:
\begin{corollary}\label{sibsPosAlphaFunctions}
Consider the same setting as in Theorem \ref{lowerBoundHolder}. For every $f:\mathcal{X}\times\mathcal{Y}\to\mathbb{R}^+$ and for every $\alpha>1$ and $\beta=\frac{\alpha-1}{\alpha}$ one has that
\begin{align}
 \mathbb{E}_{\Pm_{XY}}[f(X,Y)] \leq  &\esssup_{\Pm_Y}\left[\mathbb{E}^{\frac{1}{\beta}}_{\Pm_X}\left[f(X,Y)^\beta\right]\right]\\ &\cdot\exp\left(\frac{\alpha-1}{\alpha}I_\alpha(X,Y)\right).
\end{align} 
\end{corollary}

Differently from $I_\alpha$ with $\alpha>0$, $I_\alpha$ with $\alpha<0$ is strongly connected to converse/negative results. For instance, it can find applications in lower-bounding the Bayesian Risk in estimation procedures. In these settings one often has a non-negative loss function $\ell(\cdot,\cdot)$ that measures how far a parameter is from its estimation. The purpose is often to lower-bound the minimum expected value of this loss where the minimum is over all the possible estimators of the parameter. In such a framework Corollary \ref{sibsIAlpha} can be useful as we will see in the next section.
\section{Bayesian Risk}
%We can now immediately employ Sibson's $\alpha$-MI in providing a lower-bound on the Bayesian Risk in estimation procedures. 
Let $\mathcal{W}$ denote the parameter space and assume that we have access to a prior distribution over this space $\mathcal{P}_W$. Suppose then that we observe $W$ through the family of distributions $\mathcal{P}= \{ \mathcal{P}_{X|W=w}: w\in\mathcal{W} \}.$ Given a function $\phi:\mathcal{X}\to\hat{\mathcal{W}}$,  one can then estimate $W$ from the observations $X\sim \mathcal{P}_{X|W}$ via $\phi(X)=\hat{W}$. Let us denote with $\ell:\mathcal{W}\times\hat{\mathcal{W}}\to \mathbb{R}^+$ a loss function, the Bayesian risk is defined as:
\begin{equation}
    R = \inf_\phi\mathbb{E}[\ell(W,\phi(X)] = \inf_\phi\mathbb{E}[\ell(W,\hat{W})].\label{risk}
\end{equation}
\begin{corollary}
Consider the Bayesian framework just described, the following must hold for every $\alpha<0$ and every $\hat{W}=\phi(X^n)$:
\begin{align}
 \mathbb{E}[\ell(W,\hat{W})] 
    \geq \exp\Bigg(&\frac1\beta \Bigg(-I_\alpha(W,X^n) \notag \\ &+\log\left(\essinf_{P_{\hat{W}}} \mathbb{E}_{P_W}\left[\ell(W,\hat{W})^{\beta}\right]\right)\!\!\Bigg)\!\Bigg).  %\\
\end{align}
Moreover, taking the limit of $\alpha\to-\infty$ one recovers the following:
\begin{equation}
     \mathbb{E}[\ell(W,\hat{W})] 
    \geq \rho \exp\left(-\mathcal{L}^c(W\!\!\to \!\!X^n)\right)\!\essinf_{\Pm_{\hat{W}}}\mathbb{E}_{P_W}\!\left[\ell(W,\hat{W})\right]\!\!.
\end{equation}
\end{corollary}
\begin{proof}
The result follows from Corollary \ref{sibsIalphaBoundLess1}, Property \ref{DPI} and the fact that $W-X^n-\hat{W}$ is a Markov Chain.
\end{proof}
Another result can be included that allows us to retrieve a lower-bound that is independent of the choice of the estimator $\phi$.
\begin{corollary}
Consider the Bayesian framework just described, the following must hold for every $\rho>0$, $\alpha<0$ and every $\hat{W}=\phi(X^n)$:
\begin{align}
 \mathbb{E}[\ell(W,\hat{W})] 
    \geq \rho \exp\left(-\frac1\beta I_\alpha(W,X^n)\right)\left(1-L_W(\rho))\right)^{\frac1\beta},
\end{align}
where $L_W(\rho)=\max_{\hat{w}}\Pm_W(\ell(W,\hat{w})\leq \rho)$, also known in the literature as ``small-ball probability''.
Moreover, taking the limit of $\alpha\to-\infty$ one recovers the following:
\begin{equation}
     \mathbb{E}[\ell(W,\hat{W})] 
    \geq \rho \exp\left(-\mathcal{L}^c(W\!\!\to \!\!X^n)\right)\left(1-L_W(\rho))\right).
\end{equation}
\end{corollary}
\begin{proof}
We start by observing that for every $\rho$, by Markov's inequality
\begin{equation}
    \mathbb{E}[\ell(W,\hat{W})] \geq \rho\left(\mathbb{P}(\ell(W,\hat{W})\geq \rho)\right). \label{markov}
\end{equation}
The statement then follows from further lower-bounding $\mathbb{P}(\ell(W,\hat{W})\geq \rho)$ with Corollary \ref{sibsIalphaBoundLess1Prob} and noticing that $\min_{\hat{w}}\Pm_W(\ell(W,\hat{w})\geq \rho) = \min_{\hat{w}}(1-\Pm_W(\ell(W,\hat{w})\leq \rho))= 1-L_W(\rho).$
\end{proof}

\section{Conclusion}
In this work we extended the definition of Sibson's $\alpha$-Mutual Information to negative orders. In order to have a properly defined object and to be consistent with the axiomatic properties that an information measure should satisfy (according to R\'enyi \cite{renyiMeasuresInfo}), we slightly adapted the original definition from \eqref{sibsIAlpha} for $
\alpha>0$ to \eqref{definition} for $
\alpha<0$. We presented a sequence of properties that these objects satisfy and we connected the information-measures to a family of functional inequalities (similarly to $I_\alpha$ with $\alpha>1$ \cite[Thm 4, Remark 7, Cor 1]{fullVersionGeneralization}). This family of inequalities exists, although not always with the same inequality sign, for every $\alpha\in\mathbb{R}$ as shown in Table \ref{comparison}. When $\alpha>1$ and through these inequalities, we were able to connect the information measure to bounds on the probability of having a large generalisation error \cite{fullVersionGeneralization}, hypothesis testing problems \cite{conditionalSibsMI} and to lower-bounds on the Bayesian Risk in Estimation Procedures \cite{bayesRiskIalpha}. Here, as an example of application, we considered once again the Bayesian Risk setting and showed how this family of information measures can be employed in such a framework even if $\alpha<1$. 

\section*{Acknowledgment} The work in this paper was supported in part by the Swiss National Science Foundation under Grant 200364.
\bibliographystyle{IEEEtran}
\bibliography{sample}

\end{document}